\documentclass[11pt, peerreview]{IEEEtran}

\usepackage[dvipdfmx]{graphicx,xcolor}
\usepackage[fleqn]{amsmath}
\usepackage[varg]{txfonts}
\usepackage{color}
\usepackage{comment}
\usepackage{amsmath}
\usepackage{amssymb}
\usepackage{bm}
\usepackage{color}
\usepackage{algorithm}
\usepackage{algorithmic}
\usepackage{latexsym}

\newtheorem{definition}{Definition}

\newtheorem{theorem}{Theorem}

\newtheorem{proof}{Proof}

\begin{document}

\title{Construction of Fixed Rate Non-Binary WOM Codes based on Integer Programming}
\author{\authorblockN{Yoju Fujino and Tadashi Wadayama   } \\[0.3cm]
\authorblockA{Nagoya Institute of Technology\\
{\small email: wadayama@nitech.ac.jp }
\footnote{A part of this work was presented at the International Symposium on Information Theory and Its Applications 2016. }\\
}
}  

\maketitle
\begin{abstract}
In this paper, we propose a construction of non-binary WOM (Write-Once-Memory) codes 
for WOM storages such as flash memories.
The WOM codes discussed in this paper are
fixed rate WOM codes where messages in a fixed alphabet of size $M$
can be sequentially written in the WOM storage at least $t^*$-times.
In this paper, a WOM storage is modeled by a state transition graph.
The proposed construction has the following two features. 
First, it includes a systematic method to determine 
the encoding regions in the state transition graph.
Second, the proposed construction includes 
a labeling method for states by using integer programming.
Several novel WOM codes for $q$ level flash memories with 2 cells are constructed by the proposed construction. 
They achieve the worst numbers of writes $t^*$ that meet the known upper bound 
in the range $4 \le q \le 8, M = 8$.
In addition, we constructed fixed rate non-binary WOM codes with 
the capability to reduce ICI (inter cell interference) of flash cells.
One of the advantages of the proposed construction is its flexibility.
It can be applied to various storage devices, to various dimensions (i.e, number of cells),
and various kind of additional constraints.
\end{abstract}

%
%
%
%
%


\section{Introduction}

Recent progress of storage media 
has been creating interests on coding techniques to ensure reliability of the media and
to lengthen the life of storage media.
Write-Once-Memory (WOM) codes are getting renewed interests as one of promising 
coding techniques for storage media.
In the scenario of the binary WOM codes, 
the binary WOM storage (or channel) is assumed as follows.
A storage cell has two states 0 or 1 and the initial state is 0. If a cell changes its state to 1,
then it cannot be reset to 0 any more.  
Punch cards and optical disks are examples of the binary WOM storages.
The celebrated work by Rivest and Shamir in 1982 \cite{first-wom}
presented the first binary WOM codes and their codes induced subsequent active researches 
in the field of the binary WOM codes \cite{cohen} \cite{fiat} \cite{slc}.

A memory cell in recent flash memories has multiple levels such as 4 or 8 levels 
and the number of levels are expected to be increased further in the near future.
This trend has produced motivation to the research activities on the non-binary WOM codes that 
are closely related to the multilevel flash memories \cite{non-binary} \cite{yakkobi} \cite{q} \cite{brian}.

There are two threads of researches on the non-binary WOM codes. 
The first one is {\em variable rate codes} and the other is {\em fixed rate codes}.

The variable rate codes are the non-binary WOM codes such that
message alphabets used in a sequence of writing processes are not necessarily identical.
This means that writing rate can vary at each writing attempt.
Fu and Vinck \cite{fu} proved the channel capacity of the variable rate non-binary WOM codes.
Recently, Shpilka \cite{shpilka} proposed a capacity achieving construction of non binary WOM codes.
Moreover, Gabrys et al. \cite{non-binary} presented a construction 
of the non-binary WOM codes based on known efficient binary-WOM codes.

Although the variable rate codes are efficient because they can fully utilize the potential of a WOM storage,
fixed-rate codes that have a fixed message alphabet is more suitable for practical implementation into storage devices.
This is because a fixed amount of binary information is commonly sent from the master system to the storage.
Kurkoski   \cite{brian} proposed a construction of fixed rate WOM codes using two dimensional lattices.
Bhatia et al. \cite{lattice-wom} showed a construction of non-binary WOM codes 
that relies on the lattice continuous approximation.
Cassuto and Yaakobi \cite{q} proposed a construction of fixed rate non-binary WOM codes using lattice tiling.

Recently, fixed rate WOM code for reducing Inter-Cell Interference (ICI) is proposed by Hemo and Cassuto \cite{ici-q}.
It is known that ICI causes drift of the threshold voltage of a flash cell according to the voltages of adjacent flash cells \cite{ici-q}.
The drift of threshold voltage degrades the reliability of the flash cell and it should be avoided. 
One promising approach to reduce ICI is to use an appropriate constraint coding to avoid certain patterns incurring large ICI.
The WOM codes presented in \cite{ici-q} not only have large $t^*$ but also satisfy certain ICI reducing constraints. 

In the case of fixed rate codes, systematic constructions for efficient non-binary WOM codes 
are still open to be studied. 
Especially, perusing optimal codes with practical parameters 
is an important subject for further studies.
Furthermore, it is desirable to develop a construction of fixed rate WOM codes that have wide range of
applicability; this means that a new construction should be applicable to wide classes of 
WOM devises such as WOM devices with ICI constraint as well.

In this paper, we propose a novel construction of fixed rate non-binary WOM codes.
The target of storage media is modeled by a memory device with restricted state transitions, 
i.e.,  a state of the memory can change to another state according to a given state transition graph.
The model is fairly general and it includes a common model of multilevel flash memories.
The proposed construction has two notable features.
First, it possesses a systematic method to determine the sets called the encoding regions that 
are required for encoding processes. 
This is a critical difference between ours and the prior work using lattice tiling \cite{q} and \cite{ici-q}.
Second, the proposed construction determines an encode table used for encoding 
by integer programming.

\section{Preliminaries}

In this section, we first introduce several basic definitions and
notation used throughout the paper. 

\begin{figure}[b]
\begin{center}
\includegraphics[width = 0.45\textwidth]{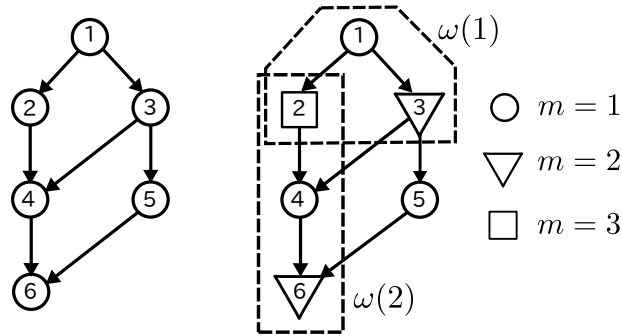}
\caption{
A state transition graph (left) and an example of encoding regions $(k=3)$ and a message function (right).
The numbers written in the nodes represent the indices of nodes.
The encoding regions $\omega(1) = \{1,2,3\}, \omega(2) = \{2,4,6\}$ are indicated by the dashed boxes (right).
The encoding regions are $\omega(x) = \emptyset$ for $x \in \{3,4,5,6\}$.
In the right figure,  the values of the message function are expressed as follows:
The values 1, 2, and 3 are represented by a circle, a triangle, and a square, respectively. 
For any message $m \in \{1,2,3\}$, 
both encoding regions $\omega(1)$ and $\omega(2)$ contain the node corresponding to $m$.
}
\label{fig:dm}
\end{center}
\end{figure}

\subsection{Basic notation}

Let $G \triangleq (V,E)$ be a directed graph
where $V = \{1,2,\ldots,|V|\}$ is the set of vertices and $E \subseteq V \times V$ is the set of edges.
If there does not exist a directed edge or a path from a vertex to itself,
then the graph $G$ is said to be a directed acyclic graph, abbreviated as DAG.
A DAG is used as a {\em state transition graph} in this paper.
The left figure in Fig.~\ref{fig:dm} is an example of DAG.
We express the DAG as $G = (V,E,r)$.
The symbol $r$ represents the root of DAG.
If for any node $s,s' \in V,(s \neq s')$ there exists the directed edge or path from $s$ to $s'$, 
we denote $s \preceq s'$.
In this case, we say that $s'$ is {\em reachable} from $s$.
  
Assume that DAG $G \triangleq (V,E,r)$.
A WOM device $D$ associated with the graph $G$ can store any $v \in V$ as its state.
The initial state of $D$ is assumed to be $r$.
We can change the state of $D$ from $s \in V$ to $s' \in V$, 
if there exists a directed edge or a path from $s \in V$ to $s' \in V$.

The message alphabet to be written in $D$ is denoted by $\mathcal{M} \triangleq \{1,2,\ldots,M\}$.
In our scenario, we want to write several messages in $\mathcal{M}$ into $D$.
Namely, a sequence of messages is sequentially written in $D$.
When we write a message $m \in \mathcal{M}$,
we must change the state of $D$.
After that, 
we get a written message $m$ in $D$ by reading the state of $D$.

\subsection{Encoding function and decoding function}

In this paper, we assume that an encoder has a unit memory 
to keep a previous input message, 
and that a corresponding decoder has no memory.

In order to write input messages into the WOM device $D$, we need an encoding function.
The definition of the encoding function is given as follows.

\begin{definition}\label{defi:enco}
Assume that a function
\[
\mathcal{E}: V \times \mathcal{M} \rightarrow V \cup \{ {\sf fail} \}
\]
is given.
The symbol {\sf fail} represents a failure of an encoding process.
If for any $s \in V$ and $m \in \mathcal{M}$, 
$
s \preceq \mathcal{E}(s,m) 
$
or 
$
\mathcal{E}(s,m) = {\sf fail},
$
then the function $\mathcal{E}$ is called an {\em encoding function}.
\end{definition}

The following definition on the decoding function 
is used to retrieve the message $m \in \mathcal{M}$ from $D$.

\begin{definition}\label{defi:deco}
Assume that a function $\mathcal{D}: V \rightarrow \mathcal{M}$ is given.
If for any $m \in \mathcal{M}$ and $s \in V$, the consistency condition
\begin{equation} \label{consistency}
\mathcal{D}(\mathcal{E}(s,m)) = m 
\end{equation} 
is satisfied,  
then the function $\mathcal{D}$ is called a {\em decoding function}.
\end{definition}

Assume that a sequence of input messages $m_1,m_2,\ldots \in \mathcal{M}$
are sequentially encoded.
We also assume that the initial node is $s_0 = r$.
The encoder encodes the incoming message $m_i$ by 
\begin{equation}
s_i = \mathcal{E}(s_{i-1}, m_i)
\end{equation}
for $i = 0, 1, \ldots$.
The output of the encoder, $s_i$, is then written into $D$ as the
next node, i.e., next state.

The following definition gives the worst number of consecutive writes to $D$ 
for a pair of encoding and decoding functions $(\mathcal{E},\mathcal{D})$.
\begin{definition}{}
Assume that a sequence of messages of length $t$, $(m_1,m_2,\ldots, m_t) \in \mathcal{M}^t$, is given.
Let $(s_1,s_2,\ldots, s_t)$ be the state sequence defined by $s_i = \mathcal{E}(s_{i-1}, m_i)$ 
under the assumption $s_0=r$.
If for any $i \in [1,t]$, $s_i \ne {\sf fail}$, then
the pair $(\mathcal{E},\mathcal{D})$ is said to be $t$ writes achievable.
The {\em worst number of writes} $t^*$ is defined by
\begin{equation}
t^* \triangleq \max \{t \mid  (\mathcal{E}, \mathcal{D}) \mbox{ is $t$-times achivable} \}.
\end{equation}
\end{definition}

In other words, the pair $(\mathcal{E},\mathcal{D})$ ensures consecutive $t^*$ writes of fixed size messages in the worst case.
Of course,  in terms of efficient use of the device $D$,  we should design $(\mathcal{E},\mathcal{D})$ to maximize $t^*$.

\section{Realization of encoding function}

In this section, 
we prepare basic definitions required for precise description 
of our encoding algorithm used in the encoding function.

\subsection{Notation}

The {\em reachable region} $R(s) (s \in V)$ is 
the set of the all nodes to which a node can change from $s$.
The precise definition is given as follows.
\begin{definition}\label{defi:reachable}
The reachable region $R(s) (s \in V)$ is defined as
\begin{equation}
R(s) \triangleq \{ x \in V | s \preceq x \}.
\end{equation}
\end{definition}

Encoding regions and a message function defined below play a critical role in an encoding process.

\begin{definition}{}\label{defi:patition}
Assume that a family of subsets in $V$, 
\[
\{\omega(s) \} \triangleq\{\omega(1), \omega(2), \ldots, \omega(|V|) \}, 
\]

is given.
Let $k$ be a positive integer satisfying $|\mathcal{M}| \le k$.
If the family satisfies the following two conditions:
\begin{enumerate}
\item $\forall s \in V, \   \omega(s) \subseteq R(s)$
\item $\forall s \in V, \   |\omega(s)| = k$ or  $\omega(s) = \emptyset$,
\end{enumerate}
then the family $\{\omega(s) \}$ is said to be the {\em encoding regions}.
\end{definition}

A message function defined below is used to retrieve a message.
\begin{definition}{}\label{defi:g}
Assume that a family of encoding regions $\{\omega(s)\} $ is given.
Let $g$ be a function: 
 \[
 g: \bigcup_{s \in V} \omega(s)  \rightarrow \mathcal{M}.
 \]
If for any $m \in \mathcal{M}$ 
and for any $s \in \{x \in V \mid  \omega(x) \ne \emptyset \}$,  there exists $a \in \omega(s)$ satisfying 
\[
g(a) = m,
\]
then the function $g$ is called a {\em message function} corresponding to the family of encoding regions $\{\omega(s) \}$.
\end{definition}

We can consider the message function as the labels attached to the nodes.
In following encoding and decoding processes, the label, i.e., the value of the message function corresponds to the message associated with the node.
This definition implies that we can find arbitrary message $m \in \mathcal{M}$ in arbitrary encoding region $\omega(s)$ $(s \in \{ x \in V \mid \omega(x) \ne \emptyset \})$.
An example of encoding regions and a message function is shown in the right-hand of Fig.~\ref{fig:dm}.

We use the above definitions of the encoding regions and the message function to encode given messages.
In order to write a sequence of messages, we must connect several nonempty encoding regions to make {\em layers}.
We here define frontiers and layers as follows. 
\begin{definition}{}\label{defi:ft}
For a subset of nodes $X \subseteq V$, 
the {\em frontier} of $X$, $F(X)$,  is defined by
\begin{equation}
F(X) \triangleq \{x \in X \mid R(x) \cap X = \{x\} \}. 
\end{equation}
\end{definition}
If $x \notin F(X)$ is hold for $x \in X$, then there exists $y \in F(X)$
which is reachable from $x$, i.e., $x \preceq y$.

A layer consists of a union set of encoding regions.
\begin{definition} \label{defi:bset}
Assume that a family of encoding regions $\{\omega(s) \}$ is given.
The {\em layer} $\mathcal{L}_i$ is recursively defined by 
\begin{equation}
\mathcal{L}_i \triangleq \bigcup_{x \in F(\mathcal{L}_{i-1})}\omega(x),\quad\mathcal{L}_0 = \{r\}.
\end{equation}
$r$ represents the root of DAG.
\end{definition}

\begin{definition}
Assume that for  integer $i \ge 0$, $\mathcal{L}_i \subset V$ is given.
The {\em start point set} $V^*$ is defined by
\begin{equation}
V^* \triangleq \bigcup_{i \ge 0} F(\mathcal{L}_i).
\end{equation}
\end{definition}

Figure \ref{fig:fe} shows an example of frontiers and layers.

\begin{figure}[tb]
\begin{center}
\includegraphics[width = 0.4\textwidth]{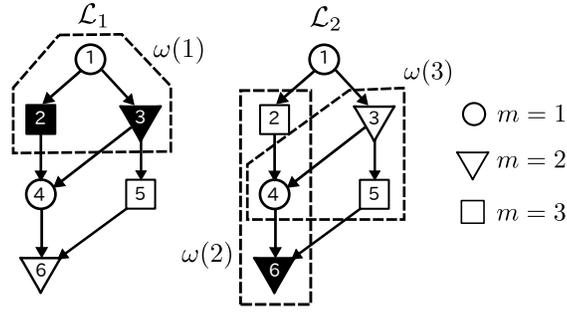}
\caption{
An example of frontiers and  layers.
The dashed boxes represent the encoding regions $\omega(1) = \{ 1,2,3\}, \omega(2) = \{2,4,6\}, \omega(3) = \{3,4,5\}$.
The encoding regions $\omega(4), \omega(5), \omega(6)$ are empty sets.
The layers are $\mathcal{L}_0 = \{1\}, \mathcal{L}_1 = \omega(1), \mathcal{L}_2 = \omega(2) \bigcup \omega(3)$.
The frontier for each layer is depicted as filled circles: $F(\mathcal{L}_1) = \{2,3\}, F(\mathcal{L}_2) = \{ 6\}$.
The start point set is 
$V^* = F(\mathcal{L}_0) \bigcup F(\mathcal{L}_1) \bigcup F(\mathcal{L}_2) = \{1,2,3,6\}$.
In the right figure,  the values of the message function are expressed as follows:
The values 1, 2, and 3 are represented by a circle, a triangle, and a square, respectively. }
\label{fig:fe}
\end{center}
\end{figure}

\subsection{Encoding algorithm}

In this subsection, we explain the encoding algorithm to realize an encoding function.
The algorithm presented here is similar to the algorithm presented in the reference \cite{q}.
The encoding algorithm is shown in Algorithm \ref{alg1}. 


\begin{algorithm}[tb]
\caption{Encoding algorithm}
\label{alg1}                          
\begin{algorithmic}[1]                  
\STATE input: $s \in V$ (current state)
\STATE input: $m \in \mathcal{M}$ (message)
\STATE output: $s'  = \mathcal{E}(s,m)$ (next state, or ${\sf fail}$)
\STATE $d:= \min [\{ x \in V \mid s \in \omega(x) \} \cup \{\infty\}] $
\IF {$\omega(d) = \emptyset$ or $d = \infty$}
\STATE{output ${\sf fail}$ and quit.}
\ENDIF
\STATE $y := \min \{x \in \omega(d) \mid g(x)=m\}$ 
\IF  {$s \preceq y$}
\STATE {$s' := y$}
\ELSE
\STATE $i =  \min\{i' | s \in \mathcal{L}_{i'} \}$
\STATE $d := \min [\{x \in F(\mathcal{L}_i)  \mid  s \preceq x\}  \cup \{\infty\} ]$
\STATE Go to line 5．
\ENDIF
\STATE output $s'$ and quit.
\end{algorithmic}
\end{algorithm}

Suppose that we have the two inputs, a state $s$ which represents the current node in the state transition graph, and a message $m$.
The main job of this encoding algorithm is 
to find $y \in \omega(d)$ satisfying $g(y) = m$ for a given message $m$. 
The encoding region  $\omega(d)$ can be considered as the current 
{\em encoding window} in which the candidate of the next state  
is found.  The variable $d$ is called 
a {\em start point} of the encoding window.
If such $y$ can be written in $D$ or is reachable from $s$ (i.e., $s \preceq y$), 
then the next state is set to $s' := y$ in line 10 of Algorithm \ref{alg1}.
Otherwise, the current encoding window should move to another encoding 
region in the next layer (line 13). 
The new start point  $z$ is chosen in the frontier $F(\mathcal{L}_i)$
and $d$ is updated as $d := z$.
\footnote{It is clear that,  for any $x \in \omega(d)$ ($\omega(d)$ is the current encoding window), 
there exists $z \in F(\mathcal{L}_i)$ satisfying $x \preceq z$.}
The layer index $i$ is the minimum index satisfying $s \in \mathcal{L}_i$.

The decoding function associated with the encoding function $\mathcal{E}$ realized by Algorithm \ref{alg1}
is given by 
\[
\mathcal{D}(x) = g(x).
\]
From the definition of the message function and the procedure of Algorithm  \ref{alg1},
it is evident that this function satisfies the consistency conditions.


We here explain an example of an encoding process
by using the state transition graph presented in Fig.~\ref{fig:fe}.
Assume that an input message sequence $(m_1, m_2) = (2,3)$ is given.
In the beginning of an encoding process, the current state is 
initialized as $s = 1$.
Since the initial message is $m_1 = 2$,   the pair $(s = 1, m = 2)$ is firstly 
given to Algorithm \ref{alg1}.


In this case, we have $d = 1$ in line 4.
Since $g(3) = 2$ is satisfied in $\omega(1)$ in line 8, 
the candidate of the next state $y = 3$ is obtained.
Because $s = 1 \preceq y = 3$ holds, we obtain $s' = y = 3$ in line 10.
The encoding process outputs $s' =3$ and then quits the process.

Let us consider the second encoding process for $m_2 = 3$.
We start a new encoding process with inputs $(s = 3, m = 3)$.
From line 4, we have $d=1$. This means that we set the encoding window to $\omega(1)$.
In this case, the encoder finds $g(2) = 3$ and lets $y = 2$. 
However, the  condition $s=3 \preceq y = 2$ is not satisfied, i.e., $y = 2$ cannot be
the next state because the node cannot change from $3$ to $2$.
In order to find the next state, we need to change the encoding window.
From line 13, the new start point of the encoding window $d = 3$ is 
chosen from the frontier as $d = \min [\{x \in F(\mathcal{L}_1) \mid s \preceq x\}] = 3$.
This operation means that we change  the encoding window from $\omega(1)$ to $\omega(3)$.
From the new encoding window $\omega(3)$, we can find $x = 5$ satisfying $g(5) = 3$.
Because $s = 3 \preceq y = 5$ holds (i.e., $y=5$ is reachable from $ s = 3 $),  we finally have the next state $s' = 5$.

\section{Construction of WOM codes}

The performance and efficiency of the WOM codes realized by the encoding and decoding functions 
described in the previous section depend on the choice of the encoding regions.
In this section, we propose a method  to create a family of the encoding regions
and a method to determine labels of nodes, i.e,  the message function by using integer programming.

\subsection{Greedy rule for constructing a family of encoding regions} \label{seq:region}
In this subsection, we propose a method for creating a family of the encoding regions 
based on a greedy rule.
The proposed WOM codes described later exploit a family of
the encoding regions defined based on the following sets.

\begin{figure}[tb]
\begin{center}
\includegraphics[width = 0.5\textwidth]{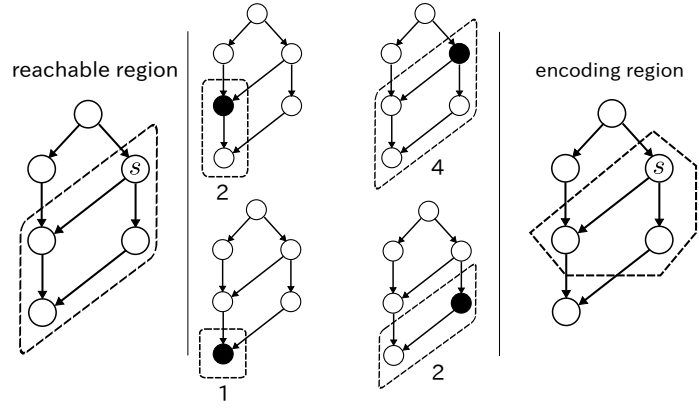}
\caption{An example of a process of the greedy construction of the encoding region for $s$, 
whose size is $k=3$.
The dashed box in the leftmost figure indicates the reachable region $R(s)$.
In the middle of the figure,
the numbers of reachable nodes for each elements in $R(s)$ are presented.
We then select top 3 nodes in terms of the number of reachable nodes
as an encoding region.
In the rightmost figure, 
the dashed box represents the encoding region $\Omega(s)$ 
constructed by the greedy process. 
} \label{fig:pro-en}
\end{center}
\end{figure}

\begin{definition} \label{defi:Omega}
Assume that an integer $k(M \le k)$ is given.
Let us denote the elements in the reachable region $R(s)$ by $r_1,r_2,\ldots,r_{|R(s)|} (s \in V)$
where the index of $r_i$ satisfies 
\begin{equation}
|R(r_1)| \ge |R(r_2)| \ge \cdots  \ge |R(r_{|R(s)| })|.  \label{eq:order}
\end{equation}
The set $\Omega(s)$ is defined by
\begin{equation}
\label{eq:omega}
\Omega(s) \triangleq 
\left\{
\begin{array}{ll}
\{ r_{1}, r_{2}, \ldots, r_{k} \}, & |R(s)| \ge k, \\
\emptyset, & |R(s)| < k.\\
\end{array}
\right.
\end{equation}
\end{definition}
In the above definition, a tie break rule is not explicitly stated.
If $|R(r_{a})| = |R(r_{b})|$ holds, we will randomly choose $r_{a}$ or $r_{b}$ to break a tie.
Figure \ref{fig:pro-en} shows an example of a greedy process for generating an encoding region.

The underlying idea in the greedy process is simply 
to enlarge future writing possibilities.
The set $\Omega(s)$ is determined by a greedy manner 
in terms of the size of reachable regions.
In other words, we want to postpone a state transition 
to a state with the smaller reachable region as late as possible.
This is because such a transition would lead to a smaller number of writes.

In the following part of this paper, we will use the encoding regions defined by
\begin{equation}
\omega(s) \triangleq 
\left\{
\begin{array}{ll}
\Omega(s), & s \in V^*, \\
\emptyset, & s \notin V^*.
\end{array}
\right.
\end{equation}

\subsection{Message labeling}\label{seq:ip}

In the previous subsection, we saw how to determine the family of the encoding regions.
The remaining task  is to find
appropriate message labels of nodes. Namely,  we must find an appropriate 
message function satisfying the required constraint described in Definition \ref{defi:g}.
In this subsection, we will propose a method to find a message function based on integer programming.

The solution of the following integer linear programming problem provides a message function.
\begin{definition} \label{defi:ip}
Assume that  a family of the encoding regions
$\{\omega(s)\}$ is given.
Let $x^*_{j,\ell}, y^*_\ell \in \{0,1\} (j \in \Gamma, \ell \in [1,k])$  
be a set of value assignments of an optimal solution 
of the following integer problem:
\begin{eqnarray}
&&{\rm Maximize} \ \sum_{\ell \in [1,k]} y_\ell \\ \nonumber
&&{\rm Subject\ to}  \\
&& \forall  i \in V^*,   \forall  \ell \in [1,k],\quad  \sum_{j \in \omega(i)} x_{j,\ell} \ge y_{\ell},  \\  \label{eq:xy}
&& \forall j \in \Gamma,\quad \sum_{\ell \in  [1,k]} x_{j,\ell} = 1,\\
&& \forall j \in \Gamma, \forall  \ell \in [1,k],  \quad  x_{j,\ell} \le y_{\ell}, \\
&& \forall j \in \Gamma, \forall  \ell \in [1,k],\quad x_{j,\ell} ,y_{\ell} \in \{0,1 \}, 
\end{eqnarray}
where $\Gamma \triangleq \bigcup_{i \ge 0} \mathcal{L}_i$.
The maximum value of the objective function is denoted by $M^*$. 
\end{definition}
The symbol $z^*_j (j \in \Gamma)$ represents 
\begin{equation}
z^*_j  \triangleq \arg \max_{\ell \in [1,k]}  \mathbb{I} [x^*_{j, \ell} = 1],
\end{equation}
where the indicator function $\mathbb{I}[condition]$ takes the value one if $condition$ is true;
otherwise it takes the value zero. If we regard $z^*_j$ as a color put on the node $j$,  the above IP problem
can be considered as an IP problem for a coloring problem. In our case, the coloring constraint is as follows:
for every node $s$ (i.e., state) in $V^*$, the neighbor of $s$ including itself contains $M^*$-colors.
This problem has close relationship to the {\em domatic partition problem}.

In the following arguments,  we set the maximize number of message $M$ equal to $M^*$.
\begin{definition}
Assume that  a function $G: \Gamma \rightarrow \mathcal{M}$ 
is defined by
$
G(j) \triangleq \alpha(z_j^*),
$
where the mapping $\alpha: A \rightarrow \mathcal{M}$ is an arbitrary bijection.
The set $A$ is defined as 
\[
A \triangleq \{\ell \in [1,k] \mid y_\ell^* = 1  \}.
\]
\end{definition}

The following theorem means that the determination of the message function can be done 
by solving the above integer programming problem.
\begin{theorem}
The function $G$ is a message function.
\end{theorem}

\begin{proof}
We assume that arbitrary $m \in \mathcal{M}$ and $i \in V^*$ are given.
First, we consider $\tilde{\ell} = \alpha^{-1}(m)$.
From the definition of the set $A$, we have $y_{\tilde \ell}^* = 1$.
The optimal solution satisfies $ \sum_{j \in \omega(i)} x^*_{j, \tilde{\ell} } \ge y_{\tilde \ell}^* = 1$.
Because $x^*_{j, \tilde{\ell} } \in  \{0,1\}$, 
there exists $\tilde{j} \in \omega(i)$ satisfying $x^*_{\tilde{j}, \tilde{\ell} } = 1$.
By the definition of the function $G$, the equation
$
 G(\tilde{j}) =  \alpha(\tilde{\ell})  =  \alpha(\alpha^{-1}(m)) = m
$
holds. This satisfies the condition for the message function.
\end{proof}

In the following, we use this message function $G$ 
in the encoding function (Algorithm 1) and the decoding function. 

\subsection{Worst number of writes}
\begin{figure}[tb]
\begin{center}
\includegraphics[width = 0.25\textwidth]{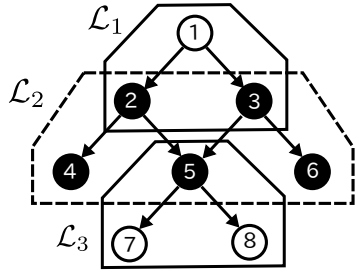}
\caption{
An example of a case where the worst number of writes is 2.
The boxes indicated in the figure is the layers $\mathcal{L}_1, \mathcal{L}_2, \mathcal{L}_3$.
The black nodes represent frontiers of layers $\mathcal{L}_1, \mathcal{L}_2$. The node with index 4 is a frontier whose encoding region 
is the empty set.  Since the node with index 4 is included in $\mathcal{L}_2$, we thus have $t^* = 2$.}
\label{fig:wn}
\end{center}
\end{figure}

The worst number of writes $t^*$ provided by 
the encoding algorithm with the encoding regions and the message function defined above 
is given by
\begin{equation}
t^* = \min\{i > 0 \mid \exists x \in F(\mathcal{L}_i), \ \omega(x) = \emptyset \}.
\end{equation}
This statement appears clear from the definition of the encoding algorithm.
Figure \ref{fig:wn} presents an example for $t^* = 2$.

\section{Numerical results on proposed WOM codes}

In this section, we will construct several classes of  fixed rate WOM codes based on the 
proposed construction. We used the IP solver IBM CPLEX for solving the integer programming problem.

\subsection{Multilevel flash memories}

Multilevel flash memories consist of a large number of cells.
Each of cell can store electrons in itself.
It is assumed that the level of a cell can be increased but cannot be decreased.
In this paper, we assume that $n$ cells that can keep $q$ level values from the alphabet $\{0,1,\ldots,q-1\}$.
The state transitions of $q$ level multilevel flash memories of $n$ cells 
can be represented by a state transition graph (directed square grid graph) presented in Fig. \ref{fig:flash}.

Figure \ref{fig:flash} presents 
the state transition graph for multilevel flash memory $(n = 2, q = 4)$ and
the encoding regions constructed by the proposed method. 
In this case, we can always write 5 messages for each write operation and 
the worst number of writes is $t^* = 2$ in this case.
\begin{figure}[b]
\begin{center}
\includegraphics[width = 0.45\textwidth]{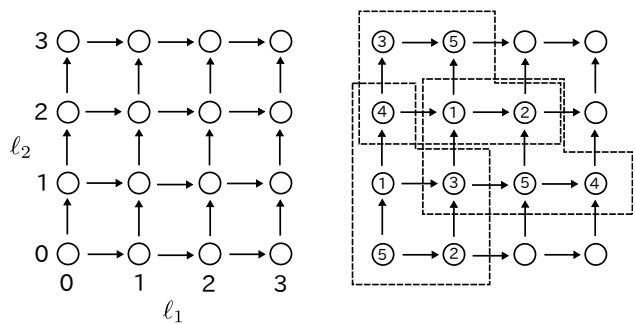}
\caption{
The left figure presents the state transition graph for multilevel flash memories $(n = 2, q = 4)$.
The levels of two cells are denoted by $\ell_1$ and $\ell_2$.
The horizontal (resp. vertical) direction means the level of the cell $\ell_1$ (resp. $\ell_2$).
The right figure presents a family of encoding regions and a message function 
constructed by the proposed method.
The numbers written in the nodes represent the values of the message function.
Nonempty encoding regions are indicated by the boxes.
For any message $m$ in $\{1,2,3,4,5\}$, 
each nonempty encoding regions contains a node corresponding to $m$. 
}
\label{fig:flash}
\end{center}
\end{figure}

\begin{table}[b]  
\caption{Worst numbers of writes $t^*$ of proposed WOM codes for $q$ level flash memories with $n = 2$ cells.}
\begin{center} 
  \begin{tabular}{c|ccccc} 
  \hline \label{tab:t-n2}
    $M \backslash q$ & 4 & 5 & 6 & 7 & 8 \\ \hline
    4 & 3 & 4 & 5 & 6 & 7\\
    5 & 2 & 3 & 4 & 5 & 6\\
    6 & 2 & 3 & 3 & 4 & 5\\
    7 & 1 & 2 & 3 & 3 & 4\\
    8 & 1 & 2 & 3 & 3 & 4\\ 
\hline
  \end{tabular}
\end{center}
\end{table}

\begin{table}[tb] 
\begin{center}
\caption{Comparison between $t^*$ of proposed WOM codes and upper bound ($q$ level flash memories)($n=2, M = 8$)} 
  \begin{tabular}{c|cccccccc} \hline \label{tab:ub}
     & $q = 4$ & 5 & 6 & 7 & 8 &16 & 32 & 48\\ \hline
    Upper bound & 1 & 2 & 3 & 3 & 4 & 9 & 20 & 31\\
    Proposed & 1 & 2 & 3 & 3 & 4 & 9 & 20 & 31\\
    \hline
  \end{tabular} 
\end{center}
\end{table}

\begin{table}[tb] 
\begin{center}
\caption{Worst numbers of writes $t^*$ of proposed WOM codes for $q$ level flash memories with $n = 3$ cells.}
  \begin{tabular}{c|ccccc} 
  \hline \label{tab:t-n3}
    $M \backslash q$ & 4 & 5 & 6 & 7 & 8 \\ \hline
    4 & 6 & 8 & 10 & 12 & 14\\
    5 & 4 & 5 & 7 & 8 & 10\\
    6 & 4 & 5 & 7 & 8 & 10\\
    7 & 3 & 5 & 6 & 8 & 9\\
    8 & 3 & 4 & 6 & 7 & 8\\ 
\hline
  \end{tabular} 
\end{center}
\end{table}

\begin{table}[tb] 
\begin{center}
\caption{Worst numbers of writes $t^*$ of proposed WOM codes for $q$ level flash memories with $n = 4$ cells.}
  \begin{tabular}{c|ccccc} 
  \hline \label{tab:t-n4}
    $M \backslash q$ & 4 & 5 & 6 & 7 & 8 \\ \hline
    5 &  7& 9 & 12 & 14 & 17\\
    6 &  5& 7 & 9 & 11 & 13\\
    7 &  5& 7 & 9 & 11 & 13\\
    8 &  5& 7 & 9 & 11 & 13\\ 
\hline
  \end{tabular} 
\end{center}
\end{table}

Table \ref{tab:t-n2} presents the worst numbers of writes $t^*$ of 
the proposed WOM codes for multilevel flash memories
for the cases $n = 2$.  
When we solved the IP problems,   $k = M$ was assumed.
For example, in the case of $q=8, M=8$, the worst number of writes equals $t^* = 4$.
In \cite{q}, several upper bounds for $t^*$ are presented for WOM codes $(n = 2, q, M, t^*)$.
For $M \ge 8$, the worst numbers of writes are upper bounded as
\begin{equation}
t^* \le \left\lceil \frac{2(q-1)}{3} \right\rceil -1. 
\end{equation}
Table \ref{tab:ub} shows 
the comparison between this upper bound and the worst numbers of writes of the proposed codes for $n=2, M = 8$.
We can see that the worst numbers of writes of the proposed WOM codes 
exactly coincide with the values of the upper bound.  
This result can be seen as an evidence of the efficiency of the WOM codes
constructed by the proposed method.

Table \ref{tab:t-n3} shows the result for $n=3$.
In \cite{q}, 
an $(n=3, q=7, M=7 ,t^*=7)$ WOM code is presented.
According to Table \ref{tab:t-n3}, the proposed WOM code attains $t^* = 8$ which is larger than that of the known code
under the same parameter setting: $n = 3, q = 7, M = 7$. 
Table \ref{tab:t-n4} shows the result for $n=4$.
In our experiments, we were able to 
construct WOM codes for the range of $M \in \{5,6,7,8\}$ and $q \in \{4,5,6,7,8\}$ with reasonable computation time. 

\subsection{WOM codes with constraints for reducing ICI}

In the current rapid grow of the cell density of NAND flash memories,  
the ICI is getting to be one of hardest obstacles 
for narrowing cell sizes.
The paper \cite{ici-q} showed several excellent fixed rate WOM codes with constraints for reducing ICI.
Their codes incorporate a constraint that keeps balance  of the charge levels of adjacent cells.
It is expected that such constraints promote a reduction on the ICI effect and leads to realizing more reliable memories. 
In this subsection, we will apply our construction to WOM codes with constraints for reducing ICI.

Assume that we have flash memory cells $c_1, c_2, \ldots, c_n$.
The current level for each cell is denoted by $\ell_i$.
The following definition gives the $d$ imbalance constraint for reducing ICI.
\begin{definition}\label{defi:imb}
Let $d$ be a positive integer smaller than $n$.
For any write sequence, if each cells $c_i (1 \le i \le n)$ satisfies 
\begin{equation}
\max_{i,j,\  i \neq j} | \ell_i - \ell_j | \le d,
\end{equation}
then we say that  the cell block satisfies {\em the $d$ imbalance constraint}.
\end{definition}
In other words,  if the cells satisfy the $d$ imbalance constraint, then the level difference between a pair of adjacent cells is limited to $d$.
It is known that a large level difference of adjacent cells tends to induce ICI. The $d$ imbalance constraint is thus helpful to reduce ICI \cite{ici-q}.

Figure.~\ref{fig:flash-ici} presents state transition graphs for 4 level flash memories of two cells $(n=2)$
with the $d$ imbalance constraint $(d = 1,2)$. From this figure,  at any state (or node),  the difference of level between $c_1$ and $c_2$
are always limited to $d (=1,2)$.
\begin{figure}[tb]
\begin{center}
\includegraphics[width = 0.47\textwidth]{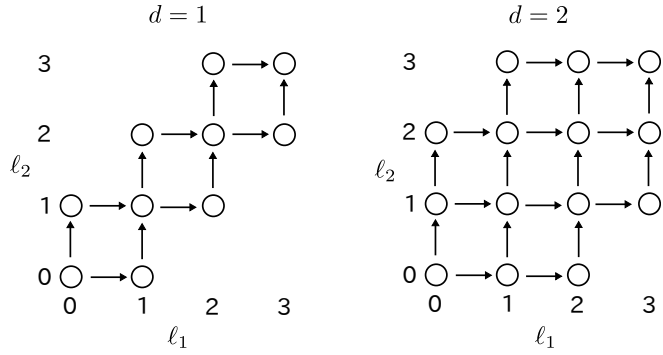}
\caption{
This figure presents a state transition diagram of
$q=4$ level flash memories of two cells $(n=2)$ satisfying 
$d$ imbalance constraints (left: $d=1$, right: $d=2$).
The levels of two cells are denoted by $\ell_1$ and $\ell_2$.
The horizontal (resp. vertical) direction means the level of the cell $\ell_1$ (resp. $\ell_2$).
}
\label{fig:flash-ici}
\end{center}
\end{figure}

It is straightforward to apply our code construction to the case of the WOM codes with $d$-imbalance constraint.
Figure \ref{fig:flash-ici-construct} presents a family of encoding functions and values of a message function constructed by the 
proposed method. This example shows universality of the proposed construction, i.e., it can be applied to any state transition graph.

\begin{figure}[tb]
\begin{center}
\includegraphics[width = 0.2\textwidth]{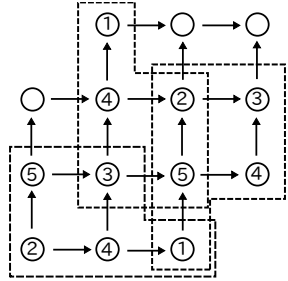}
\caption{
This state transition graph corresponds to 
$q=4$ level flash memories of two cells $(n=2)$ satisfying $d=2$ imbalance constraint.
The values in the circles represent the values of the message function.
Nonempty encoding regions are indicated by the boxes.}

\label{fig:flash-ici-construct}
\end{center}
\end{figure}

Table \ref{tab:ici-ub} shows the comparison between 
upper bound presented in \cite{ici-q} 
and the worst numbers of writes of the proposed codes $(n=2, M=8, d=3)$.
For $M = 8$, the worst numbers of writes of the WOM codes with the $d=3$ imbalance constraint 
are upper bounded by 
\begin{equation}
t^* \le \left\lfloor \frac{3(q-1)}{5} \right\rfloor.
\end{equation}
We can see that the worst numbers of writes of the proposed WOM codes 
exactly coincide with the values of the upper bound.

Tables  \ref{tab:ici-n3} and \ref{tab:ici-n4} present the worst numbers of writes $t^*$ of the proposed WOM codes with the $d$ imbalance constraint $(n = 3,4)$.
The paper  \cite{ici-q} only deals with the case of two cells $(n=2)$.
It is not trivial to construct WOM codes ($n=3, n=4$) with the $d$ imbalance constraint by using the construction given in \cite{ici-q} but
our construction is directly applicable even for such cases.

\begin{table}[tb] 
\begin{center}
\caption{Comparison between $t^*$ of proposed WOM codes with the $d$ imbalance constraint and upper bound ($n=2, M = 8, d=3$)} 
\label{tab:ici-ub}
  \begin{tabular}{c|cccccccc} \hline 
     & $q = 4$ & 5 & 6 & 7 & 8 &16 & 32 & 48\\ \hline
    Upper bound & 1 & 2 & 3 & 3 & 4 & 9 & 18 & 28\\
    Proposed & 1 & 2 & 3 & 3 & 4 & 9 & 18 & 28\\
    \hline
  \end{tabular} 
\end{center}
\end{table}

\begin{table}[tb] 
\begin{center}
\caption{Worst numbers of writes $t^*$ of proposed WOM codes with the $d$ imbalance constraint
$(n = 3)$}
  \begin{tabular}{c|cc|cc} 
  \hline \label{tab:ici-n3}
     & \multicolumn{2}{c|}{$d = 2$} & \multicolumn{2}{c}{$d = 3$} \\ \hline
     $M \backslash q$& 4 & 8 & 4 & 8 \\ \hline
    5& 4 & 10 & 4 & 10 \\
    6 & 4 & 9 & 4 & 10 \\
    7 & 3 & 9 & 3 & 9 \\
    8 & 3 & -- & 3 & 8 \\ 
\hline
  \end{tabular} 
\end{center}
\end{table}

\begin{table}[tb] 
\begin{center}
\caption{Worst numbers of writes $t^*$ of proposed WOM codes with the $d$ imbalance constraint $(n = 4)$}
  \begin{tabular}{c|cc|cc} 
  \hline \label{tab:ici-n4}
     & \multicolumn{2}{c|}{$d = 2$} & \multicolumn{2}{c}{$d = 3$} \\ \hline
     $M \backslash q$& 4 & 8 & 4 & 8 \\ \hline
    5 & 7 & -- & 7 & 17 \\
    6 & 5 & 13 & 5 & 13 \\
    7 & 5 & 13 & 5 & 13 \\
    8 & 5 & 13 & 5 & 13 \\ 
\hline
  \end{tabular} 
\end{center}
\end{table}

\section{Conclusion}

In this paper, we proposed a construction of fixed rate non-binary WOM codes 
based on integer programming.
The novel WOM codes with $n=2, M=8$ achieve 
the worst numbers of writes $t^*$ 
that meet the known upper bound in the range $q \in [4,8]$.
We discovered several new efficient WOM codes for $q$ level flash memories when $n = 3, 4$.
For instance, our  $(n=3, q= 7, M = 7, t^* = 8)$ WOM code provides a larger worst number of writes
than that of the known code with the parameters $(n=3, q= 7, M = 7, t^* = 7)$  \cite{q}.
In addition, We constructed several WOM codes with $d$ imbalance constraint for reducing ICI.
Our WOM codes with $n=2, M=8, d=3$ achieve the worst numbers of writes $t^*$ 
that meet the known upper bound in the range $q \in [4,8]$. 
This implies the efficiency of the WOM codes
constructed by our construction.
Another notable advantage of the proposed construction is its flexibility 
for handling high dimensional cases.
It is easy to construct for the codes with modestly large $n$ 
when the integer programming problem can be solved with reasonable time.
The proposed construction can be applied to various storage devices, 
to various dimensions (i.e, number of cells),  
and various kind of additional constraints.

\section*{Acknowledgment}

This work was supported by JSPS Grant-in-Aid for Scientific Research Grant Number 16K14267.


\end{document}